\documentclass[journal,final]{IEEEtran}  
\usepackage{amsmath, amsthm, amssymb}
\usepackage{epsfig}
\usepackage{latexsym}
\usepackage{graphicx}
\usepackage{cite}
\usepackage{epsfig}
\usepackage{subfigure}
\usepackage{balance}
\usepackage{color}

\newtheorem{theorem}{Theorem}
\newtheorem{lemma}{Lemma}

  \usepackage{epstopdf}

\title{Entropy and  Channel Capacity under Optimum Power and Rate Adaptation  over Generalized  Fading Conditions}

\author{
Paschalis~C.~Sofotasios,~\IEEEmembership{Member,~IEEE},~Sami~Muhaidat,~\IEEEmembership{Senior~Member,~IEEE},~Mikko~Valkama,~\IEEEmembership{Member,~IEEE}, Mounir~Ghogho,~\IEEEmembership{Senior~Member,~IEEE},~and~George~K.~Karagiannidis,~\IEEEmembership{Fellow,~IEEE}

\thanks{P. C. Sofotasios is with the Department of Electronics and Communications Engineering, Tampere University of Technology, 33101 Tampere, Finland and with the Department of Electrical and Computer Engineering, Aristotle University of Thessaloniki, 54124 Thessaloniki, Greece  \, (e-mail: p.sofotasios@ieee.org)  }

\thanks{S. Muhaidat is with the Department of Electrical and Computer Engineering, Khalifa University, PO Box 127788, Abu Dhabi, UAE and with
the Centre for Communication Systems Research, Department of Electronic Engineering, University of Surrey, GU2 7XH Guildford, U.K. (e-mail:
muhaidat@ieee.org)}

\thanks{M. Valkama is with the Department of Electronics and Communications Engineering, Tampere University of Technology, 33101 Tampere, Finland \, (e-mail: mikko.e.valkama@tut.fi)}

\thanks{M. Ghogho is with the School of Electronic and Electrical Engineering, University of Leeds,  LS2 9JT Leeds, U.K. and also with the International University of Rabat, 10100 Rabat, Morocco (e-mail: m.ghogho@ieee.org)}

\thanks{G. K. Karagiannidis is with the Department of Electrical and Computer Engineering, Aristotle University of Thessaloniki, 54124 Thessaloniki, Greece and with the Department of Electrical and Computer Engineering, Khalifa University, PO Box 127788
Abu Dhabi, UAE \, (e-mail: geokarag@ieee.org)}
}
\begin{document}

\maketitle

\begin{abstract}
Accurate fading characterization and channel capacity determination are of paramount importance in both conventional and emerging communication systems. The present work addresses the nonlinearity of the propagation medium and its effects on the channel capacity. Such fading conditions are first characterized using information theoretic measures, namely, Shannon entropy, cross entropy and relative entropy. The corresponding  effects on the channel capacity with and without power adaptation are then analyzed. Closed-form expressions are derived and validated through computer simulations. It is shown that the effects of  nonlinearities are significantly larger than those of fading parameters such as the scattered-wave power ratio, and the correlation coefficient between the in-phase and quadrature components in each cluster of multipath components.
 \end{abstract}


\section{Introduction}  \label{intro}
\IEEEPARstart{M}{ultipath} fading characterization  is considered a critical task in the effective analysis of the performance of   wireless communication systems.  To this end, the   $\alpha-\mu$,  $\eta-\mu$, $\lambda-\mu$ and   $\kappa-\mu$  distributions have been considered relatively suitable models as they also    include as special cases  the widely known Rayleigh, Nakagami$-m$, Hoyt, Weibull and Rice distributions, see  \cite{B:Alouini, Yacoub_1a, Daniel, Yacoub_2a, Yacoub_3a, Yacoub_4a, B:Sofotasios, New_6, Additional_1, Additional_2, Additional_3, Additional_4, Additional_5, Additional_6, Additional_7, Additional_8, Additional_9, C:Sofotasios_5, Boss_1, Boss_2, Boss_3} and the references therein.  \\
\indent
It is also known that accurate determination of the channel capacity  for different communication scenarios   has been of core importance in wireless communications, including requirements on optimum rate and transmit power constraints. To this end, the average channel capacity over generalized fading channels was thoroughly analyzed in \cite{Daniel}, whereas the channel  capacity under  different adaptation policies  was thoroughly   investigated in \cite{Ermolova, Bithas, Sagias}. \\
\indent
Nevertheless, in spite of the usefulness of the aforementioned  distributions, their accuracy and generality  are limited, since the  non-linearity parameter $\alpha$ in these models is not related to the fading parameters $\eta$, $\lambda$ and $\kappa$ \cite{Yacoub_1}. 
As a consequence, the detrimental effects of  non-linear propagation medium are  not considered in the   majority of   investigations. Motivated by this, the authors in \cite{Yacoub_1} proposed the $\alpha-\eta-\mu$ and  $\alpha-\kappa-\mu$ distributions which constitute a generalization of the $\alpha-\mu$, $\eta-\mu$ and $\kappa-\mu$ models and were shown to provide remarkably accurate fitting to measurement results in various communication scenarios. 
In this context, useful   statistical results and properties for the $\alpha-\lambda-\mu$ and $\alpha-\eta-\mu$ models were reported in \cite{Kotsop_2, Kotsop_3a, Kotsop_3b} whereas the corresponding outage probability was analyzed in \cite{Paschalis}. A simple and remarkably accurate  random sequence generator for $\alpha-\kappa-\mu$ and $\alpha-\eta-\mu$  variates was proposed in \cite{Cogliatti} while the corresponding symbol error rate and channel capacity was addressed in \cite{Ehab_1, Ehab_2}.    
\\
 \indent
However, none of the reported analyses address the corresponding channel capacity with optimum power and rate adaptation. Motivated by this, the aim of the present work is twofold:  we firstly derive novel analytic expressions for the Shannon entropy of  $\alpha-\eta-\mu$ and $\alpha-\lambda-\mu$ distributions. Capitalizing on this, we derive analytic expressions for the  cross-entropy and relative entropy for these distributions with respect to the $\eta-\mu$ and $\lambda-\mu$ distributions, respectively.  
Secondly, we derive novel closed-form expressions for the corresponding channel capacity under optimum rate adaptation as well as optimum power and rate adaptation. 
The derived expressions are subsequently used  in evaluating the corresponding performance, which indicates that the effects of  fading non-linearities are    larger than those of popular and widely considered  fading    parameters. 
%
 %
%
%
%
\section{ Non-Linear Fading Channels and Entropies }

\subsection{The $\alpha-\eta-\mu$ and $\alpha-\lambda-\mu$ Models}

The SNR   PDF of  the $\alpha-\eta-\mu$  distribution is given by \cite{Yacoub_1}
\begin{equation} \label{a-h-m_1}
p_{\gamma}(\gamma) = \frac{\alpha \sqrt{\pi} \mu^{\mu + \frac{1}{2}} (\eta + 1)^{\mu + \frac{1}{2}} \gamma^{\frac{\alpha \mu}{2} +\frac{\alpha}{4} - 1} I_{\mu - \frac{1}{2}} \left( \frac{(\eta^{2} - 1) \mu \gamma^{\frac{\alpha}{2}}}{2 \eta \overline{\gamma}^{\frac{\alpha}{2}}}\right)}{2 \sqrt{\eta} \Gamma(\mu) (\eta - 1)^{\mu - \frac{1}{2}} \overline{\gamma}^{\frac{\alpha \mu}{2} +\frac{\alpha}{4} }  \exp \left(\frac{(1 + \eta)^{2} \mu \gamma^{\frac{\alpha}{2}}}{2 \eta \overline{\gamma}^{\frac{\alpha}{2}}} \right)   }    
 \end{equation} 
where $\alpha$ denotes the  non-linearity of the propagation medium, $\eta$ is the scattered wave power ratio between the in-phase and quadrature components  and $\mu$ is related to the number of multipath clusters. 
 Likewise, the SNR PDF of the  $\alpha-\lambda-\mu$ fading model is expressed as  follows: 
\begin{equation} \label{a-l-m}
p_{\gamma}(\gamma) = \frac{(-1)^{\mu - \frac{1}{2}} \alpha \sqrt{\pi} \mu^{\mu + \frac{1}{2}}}{\Gamma(\mu) \lambda^{\mu - \frac{1}{2}} \sqrt{1 - \lambda^{2}}} \frac{\gamma^{\frac{\alpha \mu}{2} +\frac{\alpha}{4} - 1}}{\overline{\gamma}^{\frac{\alpha \mu}{2} +\frac{\alpha}{4} }} \frac{I_{\mu - \frac{1}{2}} \left( \frac{2 \lambda \mu \gamma^{\frac{\alpha}{2}}}{( \lambda^{2} - 1) \overline{\gamma}^{\frac{a}{2}} } \right)}{\exp \left( \frac{2 \mu \gamma^{\frac{\alpha}{2}}}{(1 - \lambda^{2}) \overline{\gamma}^{\frac{\alpha}{2}}} \right)} 
\end{equation}
where $\lambda$ denotes the correlation coefficient between the  in-phase and quadrature components of the fading signal  \cite{Kotsop_2}. 
\noindent 
It is recalled that the $\alpha-\eta-\mu$ and $\alpha-\lambda-\mu$ distributions include as special cases the $\alpha-\mu$, the $\eta-\mu$ and the $\lambda-\mu$ fading models  for $\alpha = 2$, $\eta = 2$ and $\lambda = 0$, respectively \cite{Yacoub_1}. 

\subsection{Shannon Entropy}

The  Shannon entropy is a fundamental metric which  denotes the amount of information contained in a signal and indicates the number of bits   required for encoding this signal. 
\begin{lemma}
For $\{ \alpha,  \eta, \gamma, \overline{\gamma}\} \in \mathbb{R}^{+}$ and $\mu = 1$,  the Shannon entropy of the  $\alpha-\eta-\mu$ fading distribution is expressed as
\begin{equation*}
\begin{split}
\hspace{-0.1cm}H(p) &=    \frac{\mathcal{B}_{1}}{\alpha \ln(2) } \left\lbrace  \frac{ \gamma' + \ln(\mathcal{C}_{1} + \mathcal{D}_{1})  }{ \sqrt{2 \pi \mathcal{D}_{1}} (\mathcal{C}_{1} + \mathcal{D}_{1})} - \frac{  \gamma' + \ln(\mathcal{C}_{1} - \mathcal{D}_{1})  }{\sqrt{2 \pi \mathcal{D}_{1}} (\mathcal{C}_{1} - \mathcal{D}_{1})}  \right\rbrace   \\
& \hspace{-0.04cm} +    \frac{\mathcal{B}_{1} (3 \alpha - 4) }{\ln(2) \alpha^{2}\sqrt{\pi}} \left\lbrace   \frac{  \gamma' + \ln(\mathcal{C}_{1} - \mathcal{D}_{1})  }{\sqrt{2  \mathcal{D}_{1}} (\mathcal{C}_{1} - \mathcal{D}_{1})}  -  \frac{  \gamma' + \ln(\mathcal{C}_{1} + \mathcal{D}_{1})  }{ \sqrt{2  \mathcal{D}_{1}} (\mathcal{C}_{1} + \mathcal{D}_{1})}   \right\rbrace   \\ 
\end{split}
\end{equation*}
\begin{equation}  \label{Entropy_a-h-m_1}
\hspace{0.4cm}   + \frac{\mathcal{B}_{1} \sqrt{2}}{\ln(2) \alpha  } \left\lbrace \frac{\ln(\mathcal{B}_{1}) - \ln(\sqrt{2 \pi \mathcal{D}_{1}})}{ \sqrt{\pi } \mathcal{D}_{1} ( \mathcal{C}_{1} + \mathcal{D}_{1})} + \frac{ \mathcal{D}_{1} -  \mathcal{C}_{1}\mathcal{D}_{1}^{-1}}{ \sqrt{\pi }(\mathcal{C}_{1} + \mathcal{D}_{1})^{2}} \right.  \quad \, 
   \end{equation}
   \begin{equation*}
\hspace{1.9cm}  \left.    - \frac{\ln(\mathcal{B}_{1}) - \ln(\sqrt{2 \pi \mathcal{D}_{1}})}{\sqrt{\pi } \mathcal{D}_{1}( \mathcal{C}_{1} - \mathcal{D}_{1})} - \frac{\mathcal{D}_{1} - \mathcal{C}_{1}\mathcal{D}_{1}^{-1}}{\sqrt{\pi }(\mathcal{C}_{1} - \mathcal{D}_{1})^{2}}  \right\rbrace  
\end{equation*}
in bits/message where
\begin{equation}
\mathcal{B}_{1} =  \frac{\alpha \sqrt{\pi}  (\eta + 1)^{\frac{3}{2}}}{2 \sqrt{\eta} \sqrt{(\eta - 1)} \overline{\gamma}^{\frac{3 \alpha}{4}}}
\end{equation}
%
%
%
 $\mathcal{C}_{1} = (1 + \eta)^{2}{/}(2 \eta \overline{\gamma}^{\frac{a}{2}})$ and $\mathcal{D}_{1} =  (\eta^{2} - 1) {/} (2 \eta \overline{\gamma}^{\frac{\alpha}{2}})$ with $\gamma' = 0.5772$ denoting the Euler-Mascheroni constant.  
\end{lemma}

\begin{proof}
The Shannon entropy for continuous random variables with PDF $p(x)$ is given by  $H(p) \triangleq -  \int_{0}^{\infty} p(x) \log_{2}\left(p(x) \right){\rm d}x$. 
Thus, for the case of $\alpha-\eta-\mu$ fading in \eqref{a-h-m_1} and with the aid of standard logarithmic identities it follows that 
\begin{equation}
\begin{split}
H(p) &= \frac{\mathcal{B}_{1}}{\ln(2)} \int_{0}^{\infty} \left\lbrace \frac{\mathcal{C}_{1}}{\gamma^{1 - \frac{5 \alpha}{4}}} - \frac{\ln(\mathcal{B}_{1})}{\gamma^{1 - \frac{3 \alpha}{4}}} \right\rbrace \frac{  I_{\frac{1}{2}}\left( \mathcal{D}_{1} \gamma^{\frac{\alpha}{2}} \right)}{  e^{\mathcal{C}_{1} \gamma^{\frac{\alpha}{2}}}}   {\rm d}\gamma \\
& - \frac{\mathcal{B}_{1} (3\alpha - 4)}{4 \ln(2)}  \int_{0}^{\infty}  \frac{ \gamma^{ \frac{3 \alpha}{4} -1 }   \ln(\gamma)    I_{\frac{1}{2}}\left( \mathcal{D}_{1} \gamma^{\frac{\alpha}{2}} \right)}{e^{ \mathcal{C}_{1} \gamma^{\frac{\alpha}{2}}} }   {\rm d}\gamma \\
& - \frac{\mathcal{B}_{1}  )}{  \ln(2)}  \int_{0}^{\infty} \frac{ I_{\frac{1}{2}}\left( \mathcal{D}_{1} \gamma^{\frac{\alpha}{2}} \right)  \ln\left[ I_{\frac{1}{2}} \left( \mathcal{D}_{1} \gamma^{\frac{\alpha}{2}} \right) \right]}{ \gamma^{1 -  \frac{3 \alpha}{4} }  e^{ \mathcal{C}_{1} \gamma^{\frac{\alpha}{2}}} }  {\rm d}\gamma. 
\end{split}
\end{equation}
By recalling that $I_{0.5}(x) = (e^{x} - e^{-x})/\sqrt{2 \pi x}$,  one obtains 
\begin{equation*}
H(p) =  \mathcal{B}_{1} \left[ \ln (\mathcal{C}_{1}) - \mathcal{C}_1 \right]  \int_{0}^{\infty}\frac{e^{-\frac{\gamma^{a}}{2}(\mathcal{C}_{1} + \mathcal{D}_1)} - e^{-\frac{\gamma^{a}}{2}(\mathcal{C}_{1} - \mathcal{D}_1)}}{ \ln(2) \sqrt{2 \pi \mathcal{D}_{1}} \gamma^{1- \frac{\alpha}{2}}} {\rm d}\gamma    
\end{equation*}
\begin{equation} \label{entropy_3}
\hspace{0.75cm} + \frac{\mathcal{B}_{1} (3 \alpha -4 )}{4 \ln(2) \sqrt{2 \pi \mathcal{D}_{1}}} \int_{0}^{\infty} \frac{e^{-\gamma^{\frac{\alpha}{2}} (\mathcal{C}_{1} + \mathcal{D}_{1} ) } - e^{-\gamma^{\frac{\alpha}{2}} (\mathcal{C}_{1} - \mathcal{D}_{1} ) }}{\gamma^{1 - \frac{\alpha}{2}}[\ln(\gamma)]^{-1}} {\rm d} \gamma 
\end{equation}
\begin{equation*}
\hspace{0.35cm}  + \mathcal{B}_{1}  \int_{0}^{\infty} \frac{  \gamma^{ \frac{\alpha}{2} -1 } \left\lbrace e^{-\gamma^{\frac{\alpha}{2}}(\mathcal{C}_{1} + \mathcal{D}_{1})} -   e^{-\gamma^{\frac{\alpha}{2}}(\mathcal{C}_{1} - \mathcal{D}_{1})} \right\rbrace }{  \ln\left( \frac{e^{\mathcal{D}_{1} \gamma^{\frac{\alpha}{2}}} -  e^{-\mathcal{D}_{1} \gamma^{\frac{\alpha}{2}}}}{\sqrt{2 \pi \gamma^{\frac{\alpha}{2}}}} \right)^{-1} \ln(2) \sqrt{2 \pi \mathcal{D}_1} }   {\rm d}\gamma. 
\end{equation*}  
The first two integrals in \eqref{entropy_3} can be expressed in closed-form with the aid of \cite[eq. (4.352.1)]{Tables} and \cite[eq. (8.310.1)]{Tables} and involve the digamma function, $\psi(1)$.  By also noticing that  $\exp(-x)$   in the  logarithm of \eqref{entropy_3} becomes practically negligible as $x >> 0$, recalling that  $\psi(1) = - \gamma'$ and after   some algebraic manipulations,  one obtains \eqref{Entropy_a-h-m_1}, which completes the proof. 
 \end{proof}
Importantly,  the Shannon entropy can quantify the detrimental effects of fading non-linearity by determining the variation of required bits per message for each value of $\alpha$. This demonstrates   the  difference between the $\alpha-\eta-\mu$ and $\alpha-\lambda-\mu$ distributions  with the  popular $\eta-\mu$ and $\lambda-\mu$ distributions, respectively, which is also reflected by the considered channel capacity measures. 
This is subsequently analyzed  along with the derivation of useful  analytic expressions for the corresponding   cross entropy and relative entropy metrics.

\subsection{Cross Entropy}

The cross entropy    measures the average number of bits required to encode a message when a distribution $p(x)$ is replaced by a distribution $q(x)$. In the present analysis this can exhibit the number of bits required to encode  $\alpha - \eta - \mu$ and $\alpha - \lambda - \mu$ when they are replaced  by $\eta - \mu$ and $\lambda - \mu$ distributions, respectively, which  corresponds to the case where the non-linearity parameter is not taken into account. 

\begin{lemma}
For $\{ \alpha, \eta, \eta', \gamma, \overline{\gamma}, \overline{\gamma}' \} \in \mathbb{R}^{+}$ and $\mu = 1$,  the cross entropy between   $\alpha-\eta-\mu$ and $\eta-\mu$  distributions  is given by 
\begin{equation} \label{cross_entropy_Def}
H(p, q) = \frac{\ln (\mathcal{B}_{2}) - \ln(\sqrt{2 \pi \mathcal{D}_{2}})}{ \mathcal{B}_{1}^{-1} 2^{-\frac{1}{2}} \alpha \ln(2) \sqrt{\pi \mathcal{D}_{1}}} \left\lbrace \frac{1}{\mathcal{C}_{1} + \mathcal{D}_{1}} - \frac{1}{\mathcal{C}_{1} - \mathcal{D}_{1}} \right\rbrace  \qquad \qquad  \qquad 
\end{equation}
\begin{equation*}
\hspace{1.15cm} + \frac{  \sqrt{2} (\mathcal{C}_{2} - \mathcal{D}_2)  }{    \alpha \ln(2) \sqrt{\pi \mathcal{D}_1}} \left\lbrace 
 \frac{\mathcal{B}_{1} \Gamma\left( \frac{2}{\alpha} - 1 \right)}{(\mathcal{C}_{1} - \mathcal{D}_{1})^{1 + \frac{2}{a}}} - \frac{ \mathcal{B}_{1} \Gamma\left( \frac{2}{\alpha} - 1 \right) }{(\mathcal{C}_{1} + \mathcal{D}_{1})^{1 + \frac{2}{\alpha}}} \right\rbrace
\end{equation*}
\end{lemma}

\begin{proof}
The cross entropy between two continuous distributions with PDFs $p(x)$ and $q(x)$ is given by $H(p, q) \triangleq  -\int_{0}^{\infty}p(x) \log_{2}(q(x)){\rm d}x$. 
Therefore, for the case of $\alpha-\eta-\mu$ and $\eta-\mu$ distributions, respectively,  it   follows that
\begin{equation} \label{cross_entropy}
H(p, q) = - \int_{0}^{\infty} \frac{\mathcal{B}_{1} I_{\frac{1}{2}}\left( \mathcal{D}_{1}\gamma^{\frac{\alpha}{2}}\right)}{\gamma^{\frac{1 - \alpha \mu}{2} - \frac{\alpha}{4} } e^{\mathcal{C}_{1}\gamma^{\frac{\alpha}{2}}}} \log_{2}\left( \frac{\mathcal{B}_{2}  I_{\frac{1}{2}}\left( \mathcal{D}_{2}\gamma \right)}{\gamma^{-\frac{1}{2}} e^{\mathcal{C}_{1}\gamma}}  \right) {\rm d}\gamma  
 \end{equation} 
 where  $\mathcal{B}_{2} =  \sqrt{\eta' + 1} \sqrt{\eta' - 1}  \sqrt{\pi}/(\sqrt{\eta'} \sqrt{\overline{\gamma}'})$, 
  $\mathcal{C}_{2} = (1 + \eta')^{2}/(2 \eta' \overline{\gamma'})$ and $\mathcal{D}_{2} = (\eta' + 1)(\eta - 1)/(2 \eta' \overline{\gamma'})$ correspond to the parameters of $\eta-\mu$ distribution. It is evident that the algebraic representation of \eqref{cross_entropy} is similar to  \eqref{Entropy_a-h-m_1}. Based on this, the proof follows immediately using  Lemma 1.  
\end{proof}
In the next Section, it is shown that the  value of the considered  $H(p, q)$  varies substantially at even slight variations of $\alpha$. It is also shown that this variation is larger than that from  widely used fading parameters such as $\eta$ and $\lambda$.   This justifies  
the usefulness of the $\alpha-\eta-\mu$ and $\alpha-\lambda-\mu$ distributions.

\subsection{Relative Entropy}

The relative entropy, which is also known as the Kullback-Leibler divergence,  is a measure of the distance  between two probability distributions. It  accounts for the  inefficiency of assuming that a true distribution with PDF $p(x)$ is replaced by a distribution with PDF $q(x)$.   In the present analysis, this measure quantifies  the information loss encountered when the non-linearity parameter is not taken into account, which occurs when  $\alpha-\eta-\mu$ and $\alpha-\lambda-\mu$ distributions are replaced by the simpler $\eta-\mu$ and $\lambda-\mu$ distributions, respectively. 
\begin{lemma}
For $\{ \alpha, \eta, \eta', \gamma, \overline{\gamma}, \overline{\gamma}' \} \in \mathbb{R}^{+}$,  $\mu = 1$,  the relative entropy between   $\alpha-\eta-\mu$ and $\eta-\mu$    models is expressed as
\begin{equation} \label{relative_entropy}
D(p||q) = H(p, q) - H(p)
 \end{equation} 
 where $H(p, q)$ and $H(p)$ are given in \eqref{Entropy_a-h-m_1} and \eqref{cross_entropy_Def}, respectively. 
\end{lemma}

\begin{proof}
By recalling that the relative entropy is defined as  
\begin{equation}
D(p||q) \triangleq  \int_{0}^{\infty} p(x)\log_{2}\left( \frac{p(x)}{q(x)} \right) {\rm d}x
\end{equation}
 the proof follows    with the aid of Lemma 1 and Lemma 2. 
\end{proof}

Notably, similar expressions for the Shannon entropy of $\alpha-\lambda-\mu$ distribution as well as the cross entropy and relative entropy between $\alpha-\lambda-\mu$ and $\lambda-\mu$ distributions are obtained  by setting $\eta = (1- \lambda)/(1 + \lambda)$ in \eqref{Entropy_a-h-m_1}, \eqref{cross_entropy_Def} and \eqref{relative_entropy}, respectively. Extensive results for all considered entropies for both $\alpha-\eta-\mu$ and $\alpha-\lambda-\mu$ fading conditions  are provided in Section IV, indicating the non-negligible effects of fading non-linearity.

\section{   Capacity Under Different Adaptation Policies }

\subsection{Capacity with Optimum Rate Adaptation}

The Shannon capacity of a channel constitutes a   theoretical upper bound for the maximum rate of data transmission and is also used as a benchmark for comparisons of  the channel capacity under different adaptation schemes \cite{B:Alouini}.  In what follows, we derive a closed-form expression for the average Shannon capacity over $\alpha-\eta-\mu$ fading channels. 

\begin{theorem}
For $\{ \alpha, \eta, \overline{\gamma}, B \} \in \mathbb{R}^{+}$ and $\mu \in \mathbb{N}$, the average channel capacity over $\alpha-\eta-\mu$ fading channels with optimum rate adaptation is expressed as follows: 
\begin{equation} \label{Cora_1}
C^{\rm \, \alpha-\eta-\mu}_{\rm \, ORA}  = B \sum_{i = 0}^{\mu -1} \frac{ \Gamma(\mu + i) \eta^{i} \mu^{\mu - i} ( \eta + 1)^{\mu - i} \overline{\gamma}^{\frac{\alpha(i - \mu)}{2}}}{2 \ln (2) i! \Gamma(\mu) \Gamma(\mu - i) (\eta - 1)^{\mu + i}}   \qquad \qquad \qquad 
\end{equation}
\begin{equation*}
\hspace{1.3cm}  \times \frac{G_{2\alpha, 2(1 + \alpha)}^{2(1 + \alpha), \alpha} \left(\frac{(1 + \eta)^{2} \mu^{2}}{\eta^{2}\overline{\gamma}^{\alpha}} \Big|^{\Delta \left(\alpha, \frac{i - \alpha \mu}{2}\right), \, \Delta \left(\alpha, 1 + \frac{i - \alpha \mu}{2}\right)}_{\Delta \left(\alpha, \frac{i - \alpha \mu}{2}\right), \, \Delta \left(\alpha, \frac{i - \alpha \mu}{2}\right)} \right)}{  2^{\alpha - 1} \pi^{\alpha - \frac{1}{2}}}   
\end{equation*}
where $B$ is  the corresponding bandwidth, $G(\cdot)$ is the Meijer G$-$function and  $\Delta(x, y) \triangleq  y/x, (y+1)/x, \dots, (y + x - 1)/x. $
\end{theorem}

\begin{proof}
Recalling the definition of the average ergodic capacity
\begin{equation}
 C \triangleq  B\int_{0}^{\infty} \log_{2}(1 + \gamma) p_{\gamma}(\gamma) {\rm d}\gamma
\end{equation}
it immediately follows that 
\begin{equation}\label{Cora_2}
C^{\rm \, \alpha-\eta-\mu}_{\rm \, ORA}  = B \int_{0}^{\infty} \mathcal{B}_{1}  \frac{ \ln(1 + \gamma) I_{\mu - \frac{1}{2}}\left( \mathcal{D}_{1} \gamma^{\frac{\alpha}{2}} \right)}{ \gamma^{\frac{1 - \alpha \mu}{2} - \frac{\alpha}{4} } e^{\mathcal{C}_{1} \gamma^{\frac{\alpha}{2}} }} {\rm d}\gamma. 
\end{equation}
With the aid of \cite[eq. (8.467)]{Tables}, it follows that 
\begin{equation} \label{Cora_3}
C^{\rm \, \alpha-\eta-\mu}_{\rm \, ORA}  = \sum_{i = 0}^{\mu - 1} \int_{0}^{\infty}    \frac{(-1)^{i} \mathcal{B}_{1} B \Gamma(\mu + i) \ln(1 + \gamma) e^{\mathcal{D}_{1} \gamma^{\frac{\alpha}{2}} } {\rm d}\gamma}{i! \sqrt{\pi} \Gamma(\mu - i)   (2 \mathcal{D}_{1})^{i + \frac{1}{2}} \gamma^{1 - \frac{\alpha(\mu - i)}{2}}  e^{\mathcal{C}_{1} \gamma^{\frac{\alpha}{2}} } }. 
\end{equation}
Importantly, the exponential and logarithmic terms in \eqref{Cora_3} can be expressed in terms of the Meijer G-function in \cite{C:Adamchik} yielding 
\begin{equation} \label{Cora_4}
C^{\rm \, \alpha-\eta-\mu}_{\rm \, ORA}   = \sum_{i = 0}^{\mu - 1}    \int_{0}^{\infty}  \frac{(-1)^{i} \mathcal{B}_{1} B \Gamma(\mu + i)    }{i! \sqrt{\pi} \Gamma(\mu - i)   2^{i + \frac{1}{2}}  \mathcal{D}_{1}^{i + \frac{1}{2}}  }    \qquad \qquad \qquad  \qquad \qquad 
\end{equation}
\begin{equation*}
\hspace{1.75cm}\times  \frac{ G_{2, 2}^{1,2}\left( \gamma \Big| ^{1, \, 1}_{1, \, 1} \right)  G_{0, 1}^{1,0}\left( \mathcal{C}_{1} \gamma^{\frac{\alpha}{2}} \Big| ^{ -}_{ 0} \right)}{ \gamma^{1 - \frac{\alpha(\mu - i)}{2} }}  {\rm d}\gamma.  
\end{equation*}
The above integral can be expressed in closed-form in terms of \cite[eq. (21)]{C:Adamchik}. To this effect and after long but basic algebraic manipulations, one obtains \eqref{Cora_1}. This concludes the proof. 
\end{proof}
Equation \eqref{Cora_1} is expressed in closed-form and can be computed in   software packages such as Maple and Mathematica. 

%

\subsection{Capacity with Optimum Power and Rate Adaptation}

Wireless systems are often  subject to average transmit power constraints. Based on this, the maximum capacity for optimum power and rate adaptation, $C_{\rm OPRA}$, can be determined \cite{B:Alouini}.  

\begin{theorem}
For $\{ \alpha, \eta, \overline{\gamma}, \gamma_{0}, B  \} \in \mathbb{R}^{+}$ and $\mu \in \mathbb{N}$,  the channel capacity over $\alpha-\eta-\mu$ fading channels with optimum power and rate adaptation is expressed as follows: 
\begin{equation} \label{Copra_a-h-m}
\hspace{-0.1cm} C^{\alpha-\eta-\mu}_{\rm OPRA} = 2B \sum_{k = 0}^{\mu - 1} \sum_{i = 0}^{\mu - k - 1}  \frac{(-1)^{k}\eta^{\mu}  \Gamma(\mu + k) \Gamma\left( i, \frac{(1+\eta)\mu \gamma_{0}^{\frac{\alpha}{2}}}{\eta \overline{\gamma}^{\frac{\alpha}{2}}} \right) }{k! i! \alpha \ln(2) \Gamma(\mu) (\eta - 1)^{\mu + k}} 
\end{equation}
\begin{equation*}
\hspace{1.1cm} + 2B \sum_{k = 0}^{\mu - 1} \sum_{i = 0}^{\mu - k - 1}  \frac{ (-1)^{\mu}\eta^{k} \Gamma(\mu + k) \Gamma\left( i, \frac{(1+\eta)\mu \gamma_{0}^{\frac{\alpha}{2}}}{  \overline{\gamma}^{\frac{\alpha}{2}}} \right)}{k! i! \alpha \ln(2) \Gamma(\mu) (\eta - 1)^{\mu + k}  }
\end{equation*}
where $\gamma_0$ is the SNR threshold that determines transmission. 
\end{theorem}

\begin{proof}
By firstly recalling the following generic expression
\begin{equation}
C_{\rm OPRA} \triangleq  B\int_{\gamma_0}^{\infty} \log_{2} \left(\frac{\gamma}{\gamma_0}\right) p_{\gamma}(\gamma) {\rm d}\gamma 
\end{equation}
as well as  substituting  \eqref{a-h-m_1} and expanding the logarithm  using \cite[eq. 8.467]{Tables}, one obtains \eqref{COPRA_1} (top of the next page). Notably, the two integrals that emerge from the second term of the numerator of \eqref{COPRA_1} can be expressed in terms of  \cite[eq. (8.310.1)]{Tables}  which yields
\begin{figure*}
\begin{equation} \label{COPRA_1}
\hspace{-0.1cm} C_{\rm OPRA}^{\alpha-\eta-\mu} =  \sum_{i = 0}^{\mu -1}    \int_{\gamma_0}^{\infty} \frac{\ln(\gamma)  \left[ (-1)^{i} e^{-\gamma^{\frac{\alpha}{2}}(\mathcal{C}_{1} - \mathcal{D}_1)} + (-1)^{\mu} e^{-\gamma^{\frac{\alpha}{2}}(\mathcal{C}_{1} + \mathcal{D}_1)}\right] -\ln(\gamma_0)  \left[ (-1)^{i} e^{-\gamma^{\frac{\alpha}{2}}(\mathcal{C}_{1} - \mathcal{D}_1)} + (-1)^{\mu} e^{-\gamma^{\frac{\alpha}{2}}(\mathcal{C}_{1} + \mathcal{D}_1)}\right]}{ i! \sqrt{\pi} \Gamma(\mu - i) 2^{i + \frac{1}{2}} \mathcal{D}_{1}^{i + \frac{1}{2}} B^{-1}    \mathcal{B}_{1}^{-1} \gamma^{1 - \frac{\alpha (\mu - i)}{2}} [\Gamma(\mu +  i)]^{-1}  } {\rm d}\gamma 
\end{equation}
\end{figure*}
\begin{equation*} 
\hspace{-0.1cm }  C_{\rm OPRA}^{\alpha-\eta-\mu} =   \sum_{i = 0}^{\mu -1} \frac{B (-1)^{i}    \mathcal{B}_{1} \Gamma(\mu +  i)   }{ i! \sqrt{2 \mathcal{D}_{1}\pi} \Gamma(\mu - i) 2^{i} \mathcal{D}_{1}^{i} }   \int_{\gamma_0}^{\infty} \frac{ \gamma^{\frac{\alpha (\mu - i)}{2} } \ln(\gamma)    }{ \gamma  e^{\gamma^{\frac{\alpha}{2}}(\mathcal{C}_{1} - \mathcal{D}_1)}   } {\rm d}\gamma   
\end{equation*}
\begin{equation}   \label{COPRA_1}
\hspace{1.05cm } + \sum_{i = 0}^{\mu -1} \frac{B (-1)^{\mu}    \mathcal{B}_{1} \Gamma(\mu +  i)   }{ i! \sqrt{2 \mathcal{D}_{1}\pi} \Gamma(\mu - i) 2^{i} \mathcal{D}_{1}^{i} }   \int_{\gamma_0}^{\infty} \frac{ \gamma^{\frac{\alpha (\mu - i)}{2} } \ln(\gamma)    }{ \gamma  e^{\gamma^{\frac{\alpha}{2}}(\mathcal{C}_{1} + \mathcal{D}_1)}   } {\rm d}\gamma  
 \end{equation}
\begin{equation*}
\hspace{1.05cm }  - B \sum_{i = 0}^{\mu -1} \frac{ (-1)^{i}    \Gamma(\mu +  i) \ln(\gamma_0) \Gamma\left(\mu - i, (\mathcal{C}_{1} - \mathcal{D}_{1}) \overline{\gamma}^{\frac{\alpha}{2}}\right)  }{ i! \sqrt{\pi} \alpha \Gamma(\mu - i)  \mathcal{B}_{1}^{-1} 2^{i - \frac{1}{2}} \mathcal{D}_{1}^{i + \frac{1}{2}} (\mathcal{C}_{1} - \mathcal{D}_{1})^{\mu - i} }  
 \end{equation*}
 \begin{equation*}
\hspace{1.05cm }  - B \sum_{i = 0}^{\mu -1} \frac{ (-1)^{i}    \Gamma(\mu +  i) \ln(\gamma_0) \Gamma\left(\mu - i, (\mathcal{C}_{1} + \mathcal{D}_{1}) \overline{\gamma}^{\frac{\alpha}{2}}\right)  }{ i! \sqrt{\pi} \alpha \Gamma(\mu - i)  \mathcal{B}_{1}^{-1} 2^{i - \frac{1}{2}} \mathcal{D}_{1}^{i + \frac{1}{2}} (\mathcal{C}_{1} + \mathcal{D}_{1})^{\mu - i} }. 
\end{equation*}
Integrating  by parts the two remaining integrals and applying   \cite[eq. (4.352.1)]{Tables}, the resulting integrals can be   evaluated   using \cite[eq. (8.310.1)]{Tables}. Based on this and after long but basic algebraic manipulations yields  \eqref{Copra_a-h-m}, completing  the proof. 
 \end{proof}

It is noted that closed-form expressions  for $C_{\rm ORA}^{\alpha-\lambda-\mu}$ and $C_{\rm OPRA}^{\alpha-\lambda-\mu}$ can be readily deduced by setting $\eta = (1 - \lambda)/(1 + \lambda)$ in \eqref{Cora_1} and \eqref{COPRA_1}. Also, the corresponding  channel capacities for $\eta-\mu$  and $\lambda-\mu$ fading channels are deduced for $\alpha = 2$.


%

\begin{figure}[ht]
\centering
   \includegraphics[width =9.5cm, height =6.875cm] {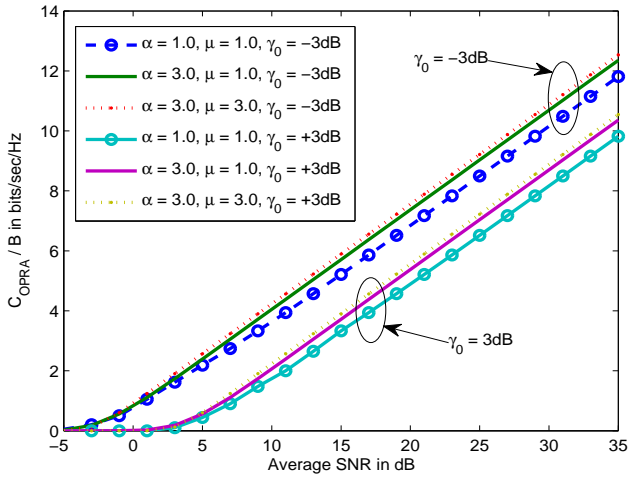}
 \label{myfigure1}
\caption{ $C_{\rm OPRA}$ vs $\overline{\gamma}$ over $\alpha-\eta-\mu$ and $\alpha-\lambda-\mu$ fading channels for $\eta = 0.6$, $\lambda = 0.25$, $\mu = 2$ and different values of the fading  non-linearity  $\alpha$.}
\end{figure}

\section {Numerical Results}
\label{Results}
The offered results are employed in quantifying the effects of the non-linearity of the propagation medium. To this end, Table I depicts the $H(p)$, $H(p, q)$ and $D(p||q)$ measures for different values of the fading parameters at $\overline{\gamma} = 15$dB.  It is shown that the Shannon entropy   is inversely proportional to the $\alpha$ parameter which indicates that  the more severe the incurred non-linearity, the more bits are required to encode the corresponding message. To this effect, it is shown that an extra bit is required  for encoding the message correctly when $\alpha = 1$, with   $\eta = 2$, $\lambda = 1/2$ and $\mu = 1$, compared to the similar case with $\alpha = 3$. On the contrary, it is observed that the lowest value of cross entropy in both cases is for $\alpha = 2$, i.e., when $\alpha-\eta-\mu$ and $\alpha-\lambda-\mu$   reduce to $\eta-\mu$ and $\lambda-\mu$, respectively. From this point, the value of $H(p,q)$ increases as the value of $\alpha$ both increases or decreases.  For example, while 7 bits/message are practically required when $\alpha = 2$, this value increases to 10 bits/message and 14 bits/message for $\alpha = 1/2$ and $\alpha = 3$, respectively. Finally, the relative entropy measure demonstrates the information lost when $\alpha$ is not considered. It is evident that  $D(p||q) = 0$  when $\alpha = 2$ and increases as $\alpha$  decreases and particularly when it increases.

The effect of fading non-linearity on $C_{\rm ORA}$ and $C_{\rm OPRA}$  is demonstrated in Table II and Fig. 1, respectively. It is shown that the corresponding spectral efficiency  is, as expected, proportional to the value of $\alpha$ in both $\alpha-\eta-\mu$ and $\alpha-\lambda-\mu$ fading conditions.  Interestingly, it is shown that the effect of $\alpha$ is more significant than the overall effect of $\eta$ and $\lambda$ parameters together with $\mu$. This illustrates the importance of considering the effects of fading non-linearity in the design and performance analysis of wireless communication systems.

\begin{table} 
\centering 
\caption{Entropies of $\alpha-\eta-\mu$ and $\alpha-\lambda-\mu$ distributions for $\overline{\gamma} = 15${\rm d}B} 
\begin{tabular}{|c|c|c|c||c|c|c|c|c|}
\hline $\alpha$-$\eta$-$\mu$&H(p)&H(p,q)&D(p$||$q)&$\alpha$-$\lambda$-$\mu$&H(p)&H(p,q)&D(p$||$q) \\ 
\hline \hline $\frac{1}{2}$-$2$-$1$&$7.03$&$9.69$&$2.66$&$\frac{1}{2}$-$\frac{1}{2}$-$1$&$6.94$&$9.42$&$2.48$   \\ 
\hline $1$-$2$-$1$&$6.87$&$8.37$&$1.50$&$1$-$\frac{1}{2}$-$1$&$6.86$&$8.21$&$1.36$  \\ 
\hline $\frac{3}{2}$-$2$-$1$&$6.56$&$6.91$&$0.36$&$\frac{3}{2}$-$\frac{1}{2}$-$1$&$6.57$&$6.87$&$0.30$     \\ 
\hline $2$-$2$-$1$&$6.28$&$6.28$&$0.00$&$2$-$\frac{1}{2}$-$1$&$6.31$&$6.31$&$0.00$     \\ 
\hline $\frac{5}{2}$-$2$-$1$&$6.05$&$7.63$&$1.59$&$\frac{5}{2}$-$\frac{1}{2}$-$1$&$6.08$&$7.52$&$1.44$     \\ 
\hline $3$-$2$-$1$&$5.83$&$13.3$&$7.45$&$3$-$\frac{1}{2}$-$1$&$5.88$&$12.5$&$6.62$     \\ 
\hline  
\end{tabular} 
\end{table}

\begin{table} 
\centering 
\caption{$C_{\rm ORA}$ for $\alpha-\eta-\mu$ and $\alpha-\lambda-\mu$ distributions} 
\begin{tabular}{||c|c||c|c||}
\hline $\alpha$-$\eta$-$\mu$ / $\overline{\gamma}$ (dB)&$C_{\rm ORA}^{\alpha-\eta-\mu} / B$&$\alpha$-$\lambda$-$\mu$ / $\overline{\gamma}$ (dB)&$C_{\rm ORA}^{\alpha-\lambda-\mu} / B$ \\ 
\hline \hline $1$-$1$-$1$ / -5dB&$0.458$& $1$-$0.1$-$1$ / -5dB&$0.358$  \\ 
\hline $1$-$1$-$1$ / 15dB&$4.432$& $1$-$0.1$-$1$ / 15dB&$4.428$  \\ 
\hline $1$-$1$-$1$ / 35dB&$10.85$&$1$-$0.1$-$1$ / 35dB&$10.85$    \\ 
\hline $2$-$1$-$1$ / -5dB&$0.529$&$2$-$0.1$-$1$ / -5dB&$0.378$   \\ 
\hline $2$-$1$-$1$ /  15dB&$4.685$&$2$-$0.1$-$1$ /  15dB&$4.674$    \\ 
\hline $2$-$1$-$1$ / 35dB&$11.25$& $2$-$0.1$-$1$ / 35dB&$11.24$     \\ 
 \hline $2$-$2$-$2$ / -5dB&$0.385$& $2$-$0.5$-$2$ / -5dB&$0.384$    \\ 
\hline $2$-$2$-$2$ /  15dB&$4.839$& $2$-$0.5$-$2$ /  15dB&$4.819$    \\ 
\hline $2$-$2$-$2$ / 35dB&$11.42$& $2$-$0.5$-$2$ / 35dB&$11.40$      \\ 
\hline $3$-$2$-$2$ / -5dB&$0.398$& $3$-$0.5$-$2$ / -5dB&$0.380$     \\ 
\hline $3$-$2$-$2$ /  15dB&$4.900$& $3$-$0.5$-$2$ /  15dB&$4.886$     \\ 
\hline $3$-$2$-$2$ / 35dB&$11.49$& $3$-$0.5$-$2$ / 35dB&$11.48$     \\ 
\hline $3$-$3$-$3$ / -5dB&$0.399$& $3$-$0.9$-$3$ / -5dB&$0.380$   \\ 
\hline $3$-$3$-$3$ /  15dB&$4.933$& $3$-$0.9$-$3$ /  15dB&$4.888$     \\ 
\hline $3$-$3$-$3$ / 35dB&$11.53$& $3$-$0.9$-$3$ / 35dB&$11.48$     \\ 
\hline    
\end{tabular} 
\end{table}


\section{Conclusion} \label{conc}

This work quantified the effects of non-linearity, $\alpha$, of the propagation medium   in wireless transmission by means of the Shannon entropy, cross entropy and relative entropy measures. The average channel capacity over generalized non-linear fading conditions was also evaluated for the cases of  optimum rate and optimum power and rate adaptations. It was   shown that the effects of  $\alpha$ are more pronounced than those of the commonly used   fading parameters $\eta$, $\lambda$ and $\mu$ combined.  

\balance
\bibliographystyle{IEEEtran}

\end{document}